\documentclass[conference, 10pt]{IEEEtran}

\usepackage{setspace}
\usepackage{amsmath}
\usepackage{amssymb}
\usepackage{mathrsfs}
\usepackage{amsthm}
\usepackage{multirow}
\usepackage{color}
\usepackage{array}
\usepackage{url}
\usepackage{comment}
\usepackage{enumerate}
\usepackage{eucal}
\usepackage[normalem]{ulem}
\usepackage{graphicx}
\usepackage{tabto}
\usepackage{amsmath}
\usepackage{caption}
\usepackage{subcaption}
\usepackage{bbm}
\usepackage{amsthm}
\usepackage{amssymb}
\usepackage{enumerate}
\usepackage{mathtools}
\usepackage{esvect}
\usepackage{caption}
\usepackage{float}
\usepackage{multicol}
\usepackage{tikz}
\usepackage{soul}
\usepackage[style=ieee]{biblatex}
\usepackage{balance}
\usepackage{adjustbox}

\usepackage{hyperref}
\newcommand\myshade{70}
\hypersetup{
	linkcolor  = red!\myshade!black,
	citecolor  = blue!\myshade!black,
	urlcolor   = blue!\myshade!black,
	colorlinks = true,
}

\usepackage{cleveref}

\usepackage{tabularx}
\usepackage{algorithm}
\usepackage[noend]{algorithmic}

\usepackage[top=0.55in, bottom=0.55in, left=0.6in, right=0.6in]{geometry}

\usepackage{tikz}
\usetikzlibrary{arrows, patterns, shapes.arrows, decorations.pathmorphing}

\IEEEoverridecommandlockouts

\theoremstyle{plain}
\newtheorem{theorem}{Theorem}

\newtheorem{lemma}[theorem]{Lemma}
\newtheorem{proposition}[theorem]{Proposition}

\theoremstyle{definition}
\newtheorem{definition}{Definition}

\newtheorem{remark}[definition]{Remark}

\newcommand{\FF}{\mathbb{F}}
\newcommand{\BB}{\mathbb{B}}

\DeclareMathAlphabet{\mathbfsl}{OT1}{ppl}{b}{it} %{OT1}{cmr}{bx}{it}

\newcommand{\vy}{\mathbfsl{y}}

\newcommand{\vzero}{\mathbf{0}}
\newcommand{\vc}{\mathbfsl{c}}

\newcommand{\vH}{\mathbfsl{H}}

\newcommand{\cA}{\mathcal{A}}

\newcommand{\cC}{\mathcal{C}}

\newcommand{\cF}{\mathcal{F}}

\newcommand{\cL}{\mathcal{L}}

\newcommand{\cT}{\mathcal{T}}

\usepackage{balance}

\newcommand{\vlambda}{{\pmb{\lambda}}}
\newcommand{\vmu}{{\pmb{\mu}}}

\newcommand{\rs}{{\rm RS}}
\newcommand{\grs}{{\rm GRS}}
\newcommand{\tr}{{\rm Tr}}

\newcommand{\floor}[1]{{\left\lfloor #1\right\rfloor}}
\newcommand{\ceil}[1]{{\left\lceil #1\right\rceil}}

\newcolumntype{Y}{>{\centering\arraybackslash}X} % for centering tables in tabularx

\title{Repair of Reed-Solomon Codes in the Presence of Erroneous Nodes\\[-3mm]}

\author{
  \IEEEauthorblockN{Stanislav Kruglik\IEEEauthorrefmark{1}, Gaojun Luo\IEEEauthorrefmark{1}, Wilton Kim\IEEEauthorrefmark{1}, Shubhransh Singhvi\IEEEauthorrefmark{2}, Han Mao Kiah\IEEEauthorrefmark{1}, San Ling\IEEEauthorrefmark{1}, Huaxiong Wang\IEEEauthorrefmark{1}}
	
 \IEEEauthorblockA{\small \IEEEauthorrefmark{1} School of Physical and Mathematical Sciences, 
		Nanyang Technological University, Singapore
    }
 \IEEEauthorblockA{\small \IEEEauthorrefmark{2} Signal Processing \& Communications Research Center, International Institute of Information Technology, Hyderabad, India
    }
  {\small \{stanislav.kruglik, gaojun.luo, wilt0002\}@ntu.edu.sg, shubhransh.singhvi@students.iiit.ac.in, \{hmkiah, lingsan, hxwang\}@ntu.edu.sg}}

\begin{document}
\date{}

\maketitle

\hspace*{-12pt}
\begin{abstract}
We consider the repair scheme of Guruswami-Wootters for the Reed-Solomon code and ask: can we correctly repair a failed node in the presence of erroneous nodes? Equivalently, we consider the collection of downloaded traces as a code and investigate its code-distance properties. We propose three lower bounds on its minimum distance and study methods to efficiently correct errors close to these bounds.  
\end{abstract}

%\todo{
%\begin{itemize}
%\item Most comments are intended for you to take note when preparing the next writeup. Comment them out when you have read and understood them.
%\item Include more citations for RECENT works that use character sums.
%\item Please see if Proposition~\ref{prop:grs-2} can apply when $\cA$ is shorter than $p^{mt}-1$. I think in this case, the distance will drop.
%\item Update Table II. Please use $\FF=GF(256)$ and $\BB=GF(16)$. Take note that the character sum bound is slightly different.

%\end{itemize}

%}

\section{Introduction}\label{sec:intro}

Distributed storage becomes more popular as data volume increases exponentially in all industries. The data can be represented as $x\in\FF^k$ for some finite field $\FF$. To protect against erasures, the data is encoded into $c = (c_1,\ldots,c_n)$ and each $c_i$ is kept in the server $i$.  One important performance metric of distributed storage is the total amount of information to download to perform the recovery, called the {\em repair bandwidth}, which was introduced in \cite{Dimakis2010}. Reed-Solomon code \cite{originalRSpaper, Dinh2022} is widely used as it allows one to recover $x$ by utilizing $k$ available servers. However, this approach is not optimal for the case of repairing a single erasure as we need to download $k$ codesymbols to recover one codesymbol.

Many studies have been conducted to improve the repair bandwidth \cite{Balaji2022,Liu2018}. 
The pioneering work of Guruswami and Wootters \cite{guruswamiwooters2017} revisits the widely used Reed-Solomon codes and shows that it is possible to decrease the repair bandwidth of single erasure dramatically when more than $k$ nodes are available. 
Roughly speaking, in this scheme, instead of downloading $k$ symbols in $\FF$, 
we download $(n-1)$ sub-symbols in a base field $\BB$ called {\em traces}.
Then using certain parity-check equations (see Section~\ref{sec::tracerepairframework} for details),
we then recover the failed node.

Later, the Guruswami-Wooters repair scheme, or the trace repair framework, was extended to different scenarios in a number of works \cite{Duursma2017, Dau2017, Kiah2018, Kiah2021, ShuttyWootters2022, Tamo2022i, Tamo2022ii, kiah2023}.
%\hm{Can we please use a more standard way of naming the bibitems? Use {\tt lastname+year}. There is a rationale. Like now, I want to arrange these citations according to year, but it is very tedious.}
%Dau et al. \cite{HM1} generalized the scheme to correct multiple erasures and there are many more works consider different parameter regimes \cite{ DM, Dau, HM2, SW, Tamo1, Tamo2}. 
All of these studies, however, assume that all available nodes give correct information. 

In this paper, we consider the case where nodes can provide wrong information and we attempt to answer the following question:  is it possible to correctly repair a node with low bandwidth in the presence of erroneous nodes? 
Previously, the problem of erroneous trace correction in Reed-Solomon repair was solved for the case of the high sub-packetization regime \cite{chenbarg2020}. 
Our approach, on the other hand, is applicable to any sub-packetization level. 
Furthermore, applications extend beyond coding for distributed storage. 
In the case of secret sharing schemes based on Reed-Solomon codes (i.e., Shamir's secret sharing scheme \cite{shamir}), our methods allow shareholders to successfully obtain a secret in the presence of malicious shareholders.

\subsection{Motivating Example}

Let $\FF=GF(16)$. Set $n=16$ and $k = 2$, and we consider a $[16,2,15]-$RS code. 
So, we can correct five errors and four erasures. 
Hence, in the classical approach, when there are at most \textcolor{black}{five erroneous nodes}, 
we download $(16-4)=12$~symbols from any twelve available nodes to repair a failed node.
In other words, the repair bandwidth in this case is $12(4)=48$ bits. 

On the other hand, we consider $\FF$ as an extension field of $\BB=GF(4)$.
Then Guruswami-Wooters repair scheme~\cite{guruswamiwooters2017} allow us to repair a failed node by downloading 15 traces (see Theorem~\ref{thm:GW}).
Later, we show that the traces \textcolor{black}{form a $\BB$-linear code} with minimum distance 11 (see Theorem~\ref{thm:main}). 
Therefore, using these 15 traces, we are able to correct five errors. 
Here, our repair bandwidth is $15(2)=30$ bits only.

\subsection{Our Contributions}
In the spirit of Guruswami and Wootters, our primary objective is to simply understand what can be done for Reed-Solomon codes. Specifically, we focus on the \textcolor{black}{Guruswami-Wootters} repair scheme (which we review in Section~\ref{sec::tracerepairframework}) and ask: can we correctly repair a failed node in the presence of erroneous nodes? 
Equivalently, we consider the collection of traces as a {\em $\BB$-linear code $\cT$} and ask what is the minimum distance of this code. 
In Section~\ref{sec:lowerbounds}, we first show this code is in fact a subcode of a generalized Reed-Solomon code. Hence, we are able to apply efficent decoding algorithms like Berlekamp-Welch to correct errors. This gives us a lower bound when $k$ is small. 
For larger values of $k$, we construct additional parity check equations in $\FF$ and use {\em lifted decoding}  to correct errors (see Section~\ref{sec:lifted-decoding}).
Finally, we use the {\em character sums} to provide another lower bound. 
%This bound works only for full-length Reed-Solomon codes with moderate values of $k$ with base field $\FF_p$ for some prime $p$. 
We remark that similar techniques were used in \cite{roth2006, li14, ding15, tang16, luo21, luo22}, \textcolor{black}{but most of these works focus on polynomial trace codes, while we consider more general rational trace codes}.
To efficiently correct errors close to these bounds, we modify the famous Guruswami-Sudan list-decoding algorithm in Section~\ref{GS}. Finally, in Section~\ref{NR}, we compare the various bounds obtained in the paper.
%\hm{Unless we have interesting open problems, there is no need to tell people what is in the conclusion. \sout{and in Section~\ref{Conclusion} we provide the conclusion}.}

\vspace{2mm}
\section{Preliminaries}
\vspace{2mm}
Let $[n]$ denote the set of integers $\{1,\ldots,n\}$. 
Let $\BB$ be the finite field with $p^{m}$ elements %for some prime $p$ and integer $m\ge 1$. 
and $\FF$ be the extension of degree $t\ge 1$. 
So, $|\FF|=p^{mt}$ and $|\BB|=p^m$.
We refer to the elements of $\FF$ as \textit{symbols} and to the elements of $\BB$ as \textit{sub-symbols}.
We use $\FF[x]$ to denote a ring of polynomials over finite field $\FF$. 

An $\FF$-linear $[n, k]$ code $\cC$ is $k$-dimensional subspace of $\FF^n$. 
We denote the dual of code $\cC$ by $\cC^{\perp}$ and 
so, for each $\vc=(c_1,\ldots,c_n)\in\cC$ and $\vc^{\perp}=(c_1^{\perp},\ldots,c_n^{\perp})\in\cC^{\perp}$, it holds that $\sum_{i=1}^nc_ic_i^{\perp}=0$.
%\hm{Please use the macro {\tt vc} for boldsymbols.}
We denote the minimum distance of $\cC$ with $d(\cC)$ and the \textcolor{black}{Singleton bound} states that $d(\cC)\leq n-k+1$ (see for example~\cite{roth2006}). 
%\hm{Two things: Singleton bound holds for ALL codes (not just linear).Next, in the references, just have one textbook. }
%Minimum Hamming distance between two distinct codewords of code $\cC$ is the minimum distance of code $\cC$. We denote is as $d(\cC)$ or, for simplicity, $d$. For any linear code the Singleton bound holds and states that $d\leq n-k+1$ 
Codes that \textcolor{black}{attain} this bound are called maximum-distance separable (MDS) codes and 
in this work, we focus on the following class of MDS codes.

\begin{definition}\label{RS}
%Let $\FF[x]$ denote a ring of polynomials over finite field $\FF$
Let $\cA\subseteq\FF$. 
The {\em Reed-Solomon} code $\rs(\cA,k)$ of dimension $k$ with evaluation points $\cA$ is defined as:
\begin{equation*}%\label{eq:RS}
\rs(\cA,k) \triangleq \{(f(\alpha))_{\alpha\in\cA}: f\in\FF[x], \deg(f(X))\le k-1\}\,,    
\end{equation*}
while the {\em generalized Reed-Solomon} code $\rs(\cA,k)$ of dimension $k$ with evaluation points $\cA\subseteq\FF$ and multiplier vector $\vlambda\in\FF^n\setminus\{\mathbf{0}\}$ is defined as:  
{\small
\begin{equation*}%\label{eq:GRS}
\grs(\cA,k,\vlambda) \triangleq \{(\lambda_\alpha f(\alpha))_{\alpha\in\cA}: f\in\FF[x], \deg(f(X))\le k-1\}\,.
\end{equation*}
}
%\hm{Please use {\tt em} to emphasize the terms you are defining and remember your periods.}
\end{definition}

Clearly, the generalized Reed-Solomon code with multiplier vector $\vlambda=(1,\ldots,1)$ is \textcolor{black}{a Reed-Solomon code} of the same length and dimension. It is well known (see~\cite{roth2006}) that dual of $\rs(\cA,k)$ is $\grs(\cA,|\cA|-k,\vlambda)$ for $\vlambda=(\lambda_{\alpha})_{\alpha\in \cA}$ where
\begin{equation}
    \lambda_j=\frac{1}{\prod_{\alpha_i\in\cA\setminus\{\alpha_j\}}(\alpha_j-\alpha_i)}.
\end{equation}
Note that when $\cA=\FF$, we have $\lambda_{\alpha}=1$ for all $\alpha\in\cA$. If it is clear from the context, we use $f(x)$ to denote the polynomial of degree at most $k-1$ corresponding to $\rs(\cA,k)$ and $r(x)$ to denote the polynomial of degree at most $|\cA|-k-1$ corresponding to the dual codeword in $\cC^{\perp}$.

\subsection{Trace Repair Framework}\label{sec::tracerepairframework}
In this section, we discuss about trace repair framework to recover a single erased node. The main idea of trace repair framework is that we want to recover a symbol in $\FF$ by using sub-symbols in $\BB$. Without loss of generality, let us assume that $f(0)$ is erased. Let $\cA\subseteq \FF\setminus\{0\}$ be the set of evaluation points.

We consider trace function $\tr:\FF\to\BB$ defined as
\begin{align}
    \tr(x) = \sum_{i=0}^{t-1} {x^{|\BB|}}^i,\quad\text{for all }x\in\FF.
\end{align}
Clearly, $\tr(x)$ is a polynomial in $x$ with degree $p^{mt-m}$. Next, we discuss how this trace function helps us in the recovery. %As $\FF$ is an extension field of $\BB$, 
We regard $\FF$ as a $\BB$-linear vector space of dimension $t$ and 
let $\{u_1,\ldots,u_t\}$ be a basis of $\FF$ over $\BB$. 
Furthermore, there exists a  {\em trace-dual basis} $\{\widetilde{u}_1,\ldots,\widetilde{u}_t\}$ for $\mathbb{F}$ such that $\tr(u_i\widetilde{u}_j)=1$ if $i=j$, and $\tr(u_i\widetilde{u}_j)=0$, otherwise. The following result plays a crucial role in our evaluation framework.

\begin{proposition}[{\cite[Ch. 2]{Lidl1996}}]\label{prop:trace}
	Let $\{u_1,\ldots,u_t\}$ be a $\BB$-basis of $\mathbb{F}$. Then there exists a trace-dual basis  $\{\widetilde{u}_1,\ldots,\widetilde{u}_t\}$ and we can write each element $x\in\mathbb{F}$ as
	\begin{equation*}
		x=\sum_{i=1}^t\tr(u_ix)\widetilde{u}_i.
	\end{equation*}
% \hm{Please learn to use the {\tt theorem} environment correctly. 
% I am referring to the label in the parenthesis. If used correctly, the label is NOT in italics.}
\end{proposition}

This means, in order to recover $f(0)$, we need to determine $\tr(u_i\lambda_0f(0))$ for all $i=1,\ldots,t$ by downloading certain information from the remaining nodes. To do so, we consider
\begin{equation}
    p_i(x) = \frac{\tr(u_ix)}{x},\quad\text{for all }i=1,\ldots,t.
\end{equation}
We can check that $p_i$ is a polynomial of degree $p^{mt-m}-1$ and $p_i(0) = u_i$. If $p^{mt-m}-1 \le |\cA| - k$, then the following parity check equations hold
\begin{align}\label{eq:paritycheckGW}
    u_i\lambda_0 f(0) = \sum_{\alpha\in\cA} p_i(\alpha)\lambda_{\alpha}f(\alpha).
\end{align}
Applying trace function both sides of \eqref{eq:paritycheckGW}, we obtain
\begin{align}
    \tr(u_i\lambda_0 f(0)) &= \sum_{\alpha\in\cA} \tr(p_i(\alpha)\lambda_{\alpha}f(\alpha))\notag\\
    &= \sum_{\alpha\in\cA} \tr(u_i\alpha)\tr\left(\frac{\lambda_\alpha f(\alpha)}{\alpha}\right). \label{eq:trace-repair}
\end{align}

Therefore, it suffices to download $\tr(\lambda_\alpha f(\alpha)/\alpha)$ from node $\alpha$.
This motivates us to study the following code.

\begin{definition}
    The {\em repair-trace} code with evaluation points $\cA\subseteq \FF\setminus\{0\}$ is defined as:
    \begin{align}
        \cT(\cA,k) & \triangleq \{(\tr(\lambda_\alpha f(\alpha)/\alpha))_{\alpha\in\cA} :f\in\FF[x], \notag \\ 
        & \hspace{30mm} \deg(f(X))\le k-1\}\,.
    \end{align}
\end{definition}

\begin{remark}
It is possible that $|\cT(\cA,k)|< |\FF|^{k}$. That is, in the definition of $\cT$, it is possible for two distinct polynomials $f$ and $g$ (with degrees at most $k-1$) to correspond to the same codeword. In other words, $\tr(\lambda_\alpha f(\alpha)/\alpha)=\tr(\lambda_\alpha g(\alpha)/\alpha)$ for $\alpha\in \cA$.
Nevertheless, $\cT(\cA,k)$ is a $\BB$-linear code.
\end{remark}

The above technique is summarised into the theorem below.

\begin{theorem}[{\cite[\textcolor{black}{Guruswami-Wootters}]{guruswamiwooters2017}}]\label{thm:GW}
    If $|\cA|\ge p^{mt-m}-1+k$ and given $\vc \in \cT(\cA,k)$, we can efficiently compute $f(0)$.
\end{theorem}
Our main task is to determine the minimum distance of $\cT(\cA,k)$. If the distance is $d$, then we are able to correct $\floor{(d-1)/2}$ errors. 
In addition, we also investigate algorithms that are able to correct the number of errors efficiently.
\begin{comment}
\begin{itemize}
	\item Without loss of generality, assume that node $0$ is failed.
	\item Let $\cA \subseteq \FF \setminus \{0\}$.
	\item Define $\cT(\cA,k) = \{(\tr(\lambda_\alpha f(\alpha)/\alpha))_{\alpha\in\cA} : \deg(f(X))\le k-1\}$. Let us call this a {\em repair-trace code}.
	\item Remark that it is possible that $|\cT(\cA,k)|<p^{mk}$. That is, in the definition of $\cT$, it is possible for two distinct polynomials $f$ and $g$ (with degrees at most $k-1$) to correspond to the same codeword. In other words, $\tr(\lambda_\alpha f(\alpha)/\alpha)=\tr(\lambda_\alpha g(\alpha)/\alpha)$ for $\alpha\in \cA$.
	\item {\bf Theorem} (Guruswami-Wooters). If $|\cA|\ge p^{mt-m}-1+k$ and given $\vc \in \cT(\cA,k)$, we can efficiently compute $f(0)$.
	\item Our main task: determine the minimum distance of $\cT(\cA,k)$. If the distance is $d$, then we are able to correct $\floor{(d-1)/2}$ errors. In addition, in this paper, we also investigate algorithms that are able to correct the number of errors efficiently.
\end{itemize}
\end{comment}

\subsection{Main Results}

In this conference paper, to simplify our exposition, we focus on the case where the data is stored on a full-length Reed-Solomon code \textcolor{black}{of length $n=p^{mt}$, dimension $k$ and code rate $R$}. Hence, $\lambda_\alpha = 1$ for all $\alpha$.
Then the repair-trace code is simply $\cT(\cA,k) = \{(
\tr(f(\alpha)/\alpha))_{\alpha\in\cA} : \deg(f(X))\le k-1\}$ and we summarize our results in  the following theorem.

\begin{theorem}\label{thm:main}
Consider the full-length Reed-Solomon code and let $0$ be the failed node $f(0)$. 
Let $d$ be the minimum distance of corresponding repair-trace code $\cT(\cA,k)$ \textcolor{black}{with $\cA=\FF\setminus\{0\}$}. 
The following bounds on $d$ hold:
\begin{enumerate}[(i)]
    \item (Degree Bound). \textcolor{black}{If $k\leq p^m$, then $d\ge p^{mt}-1-\Delta\triangleq d_1$, where}
    %If $|\cA|\geq\Delta$, then $d\ge |\cA|-\Delta\triangleq d_1$, where
\begin{equation}\label{eq:Delta}
    \Delta \triangleq 
    \begin{cases}
        (k-1)p^{mt-m}, &\text{ when }k\ge 2,\\
        p^{mt-m}-1, &\text{ when }k= 1.
    \end{cases}
\end{equation}
    \item (Lifted Decoding). \textcolor{black}{If $k\leq p^{mt}-p^{mt-m}$, then $d\ge \lfloor\frac{p^{mt}-k}{p^{mt-m}}\rfloor \triangleq d_2 $.}
     %If $|\cA|\ge p^{mt} - \lfloor\frac{p^{mt}-k}{p^{mt-m}}\rfloor$, then $d\ge \lfloor\frac{p^{mt}-k}{p^{mt-m}}\rfloor - (p^{mt}-1-|\cA|) \triangleq d_2 $. 

    \item (Character Sum Bound). \textcolor{black}{If $k<1+\frac{p^{mt}-1} {\sqrt{p^{mt}}}$,then $d\ge d_3$, where}
    %If $\cA=\FF\setminus\{0\}$ and $k<1+\frac{p^{mt}-1} {\sqrt{p^{mt}}}$, then $d\ge d_3$, where 
    \begin{equation}\label{eq:d3}
    \footnotesize
        d_3 \triangleq 
        \begin{cases}
            \frac{p^m-p}{p^m}\left(p^{mt}-1-(k-1)\sqrt{p^{mt}}\right), &\text{ when }m\ge 2,\\
            \frac{p-1}{p}\left(p^{t}-1-(k-1)\sqrt{p^{t}}\right), &\text{ when }m = 1.
        \end{cases}
    \end{equation}
    %$m=1$, $d\ge\frac{p-1}{p}\left(p^{t}-1-(k-1)\sqrt{p^{t}}\right) \triangleq d_3$.
\end{enumerate}
\end{theorem}

We can efficiently correct up to the distances promised by Theorem~\ref{thm:main}(i) and (ii).  
For Theorem~\ref{thm:main}(iii), we modify the famous Guruswami-Sudan algorithm to correct errors close to the character sum bound. \textcolor{black}{We do note that results from Theorem~\ref{thm:main}(i) and (ii) can be generalized for non-full-length Reed-Solomon codes. }

\section{Lower Bounds for Minimum Distance}\label{sec:lowerbounds}

In this section, we prove Theorem~\ref{thm:main}.
Recall that $\cA$ is a set of nonzero elements in $\FF$. First, we consider the code $\cT_1 \triangleq \left\{\left(\alpha^{p^{mt-m}} c_{\alpha}\right)_{\alpha\in\cA} : \vc \in \cT(A,k)\right\}$
and show that it is a subcode of some generalized Reed-Solomon code.

\begin{proposition}\label{prop:grs-1}
Let $\Delta$ be defined in~\eqref{eq:Delta}. If $\Delta<|\cA|$, then
$\cT_1\subseteq\grs(\cA,\Delta+1,\vmu_1)$ for some $\vmu_1$ 
\end{proposition}

\begin{proof}
We note that $\vc^{*}\in\cT_1$ can be represented as 
\begin{align*}
(F(\alpha))_{\alpha\in\cA}&=\left(\alpha^{p^{mt-m}}\tr(\frac{f(\alpha)}{\alpha})\right)_{\alpha\in\cA}\\
&=\left(\alpha^{p^{mt-m}}\sum_{i=0}^{t-1}\left(\frac{f(\alpha)}{\alpha}\right)^{p^{mi}}\right)_{\alpha\in\cA}\\
&=\left(\alpha^{p^{mt-m-1}}f(\alpha)+\cdots+f(\alpha)^{p^{mt-m}}\right)_{\alpha\in\cA},
\end{align*}
where $F(x)$ is a polynomial of degree $\max(k+p^{mt-m}-2,\ldots,(k-1)p^{mt-m})$. This fact finishes the proof. 
%\hm{When using {\tt align}, align the equality sign! Use {\tt cdots} when between operators.}
\end{proof}

Since $\cT_1$ is equivalent to the repair-trace code $\cT(\cA,k)$ (since $\alpha\ne0$ for all $\alpha\in\cA$), the minimum distance of $\cT(\cA,k)$ is at least $|\cA|-\Delta$ and we obtain Theorem~\ref{thm:main}(i). We note that we can efficiently correct up to the promised distance using any Reed-Solomon code bounded-distance decoder (see \cite{roth2006}).

\subsection{Lifted Decoding}\label{sec:lifted-decoding}
%Proposition~\ref{prop:grs-1} applies only when $\Delta< |\cA|$. In other words, when $k<p^{m}$. In this section, we study the other extreme and look for bounds that apply when $k$ is large.
%In fact, when $\cA=\FF\setminus\{0\}$, we demonstrate that $\cT(\cA,k)$ is a subcode of a generalized Reed-Solomon code with a different set of parameters. To this end, we form the following set of parity-check equations similar to \eqref{eq:trace-repair}.
\textcolor{black}{Proposition~\ref{prop:grs-1} applies only when $k$ is at most $\sqrt{n}$. In this section, we study the other extreme and look for bounds that apply when $k$ is large. In fact, we demonstrate that $\cT(\cA,k)$ is a subcode of a generalized Reed-Solomon code with a different set of parameters and sub-constant minimum distance at least $p^{m}(1-R)$. To this end, we form the following set of parity-check equations similar to \eqref{eq:trace-repair}.}

\begin{lemma}\label{lem:parity}
	For $2\le \ell\le \floor{\frac{p^{mt}-k}{p^{mt-m}}}$, we have that
	\[\sum_{\alpha\in\cA} \alpha^\ell \tr\left(\frac{f(\alpha)}{\alpha}\right) = 0. \]
\end{lemma}
%\hm{Please use {\tt ell}.}

\begin{proof}
Let $\{u_1,\ldots,u_t\}$ be the basis of $\FF$ over $\BB$ and $\{\eta_1,\ldots,\eta_t\}$ be its trace-dual basis. 
We have the following codewords of the dual of $\rs(\cA\cup\{0\},k)$: 
\begin{equation}
    r_i^{(\ell)}(x)\triangleq \frac{\textrm{Tr}\left(u_ix^\ell\right)}{x},
\end{equation}
for all $i=1,\ldots,t$ and $\ell=2,\ldots,\lfloor\frac{p^{mt}-k}{p^{mt-m}}\rfloor$. 

It is clear that $r_i^{(\ell)}(0)=0$ and the polynomial $r_i{(\ell)}(x)$ is of degree at most $\ell p^{mt-m}-1\leq p^{mt}-k-1$ for all $i$ and $\ell$. 
Then we have the following parity-check equations for code $\rs(\cA\cup\{0\},k)$. 
\begin{equation*}
r_i^{(\ell)}(0)f(0)+\sum_{\alpha\in\cA}r_i^{(\ell)}(\alpha)f(\alpha)=0.    
\end{equation*}

Following the definition of $r_i^{(\ell)}(x)$, we have
\begin{equation}\label{eq:lifteddec1}
\sum_{\alpha\in\cA}f(\alpha)\frac{\tr(u_i\alpha^\ell)}{\alpha}=0.    
\end{equation}

Applying the trace function to both sides of~\eqref{eq:lifteddec1} and employing the fact that $\tr(a\tr(b))=\tr(b\tr(a))=\tr(a)\tr(b)$ we have
\begin{equation*}
\sum_{\alpha\in\cA}\tr\left(u_i\alpha^\ell\tr\left(\frac{f(\alpha)}{\alpha}\right)\right)=0.
\end{equation*}

Utilizing the linearity of trace function, we have
{\small 
\begin{equation*}
\tr\left(\sum_{\alpha\in\cA}u_i\alpha^l\tr\left(\frac{f(\alpha)}{\alpha}\right)\right)=\tr\left(u_i\sum_{\alpha\in\cA}\alpha^\ell\tr\left(\frac{f(\alpha}{\alpha}\right)\right)=0\,.
\end{equation*}
}
Consequently,
\begin{equation*}
\sum_{i=1}^t\eta_i\tr\left(u_i\sum_{\alpha\in\cA}\alpha^\ell\tr\left(\frac{f(\alpha)}{\alpha}\right)\right)
=\sum_{\alpha\in\cA}\alpha^\ell\tr\left(\frac{f(\alpha)}{\alpha}\right)=0.
\end{equation*}
\end{proof}

Then the following proposition is immediate from Lemma~\ref{lem:parity}.

\begin{proposition}\label{prop:grs-2}
	$\cT(\cA,k) \subseteq \grs(\cA,p^{mt}-\lfloor\frac{p^{mt}-k}{p^{mt-m}}\rfloor,\vmu_2)$ for some multiplier $\vmu_2$. %$\vmu_2\in\FF_q^{|\cA|}$.
\end{proposition}
\begin{proof}
From Lemma~\ref{lem:parity}, it is clear that parity-check matrix of code $\cT$ has the following rows
\begin{equation*}
  \vH = \begin{bmatrix}
    \alpha_1^2 & \alpha_2^2 & \cdots & \alpha_n^{2} \\
    \alpha_1^{3} & \alpha_2^3 & \cdots & \alpha_n^{3} \\
    \vdots & \vdots && \vdots \\
    \alpha_1^{\ell} & \alpha_2^{\ell} &\cdots & \alpha_n^{\ell} \\
  \end{bmatrix} \,.
\end{equation*}
for $\ell=\lfloor\frac{p^{mt}-k}{p^{mt-m}}\rfloor$ and $\cA=\{\alpha_1,\ldots,\alpha_n\}$.
It is clear that the dual of the code generated by $\vH$ is a ${\rm GRS}(\cA,|\cA|-\ell+1,\vmu_2)$ for some multiplier $\vmu_2$. Therefore $\cT(\cA,k)$ is a subcode of the latter and we obtain the proposition.
\end{proof}

\textcolor{black}{Therefore, every nonzero codeword in $\cT(\cA,k)$ has weight at least $\lfloor\frac{p^{mt}-k}{p^{mt-m}}\rfloor$ and statement of Theorem~\ref{thm:main}(ii) follows.}

%The GRS code identified in the proposition above can be punctured to any subset $\cA'$ retaining the dimension and MDS property. The latter means that the code $\cT(\cA', k)$ has minimum distance at least $\lfloor\frac{p^{mt}-k}{p^{mt-m}}\rfloor-(p^{mt}-1-|\cA'|)$ and hence, we obtain Theorem~\ref{thm:main}(ii).

\subsection{Character Sum Bound}\label{sec:character-sum}
%\hm{Please use sentence case for section and subsection headings. }
In this subsection, we prove the Theorem~\ref{thm:main}(iii) by modifying the proof of \cite[Theorem~5.4]{roth2006} for two cases, $m=1$ and $m>1$. Before we proceed further, let us provide a short overview of character sums and refer the reader to \cite{Lidl1996, pascale2013} for more details. Assume that $\omega=e^{\frac{2i\pi}{p}}$ is a primitive $p$-th root of complex unity. 
It is well known that for any $x\in\BB$ it holds that 
\begin{equation}\label{eq:complexroot2}
\sum_{a\in\BB\setminus\{0\}}\omega^{ax} = 
\begin{cases}
  p-1 & \text{if $x=0$} \\
  -1 & \text{otherwise}.
\end{cases}  
\end{equation}

For any element $a$ from $\FF$, we can define an additive character as function $\chi_a(x)=\omega^{{\rm Abs}\tr(ax)}$, where $x\in\FF$ and ${\rm Abs}\tr(\cdot)$ is the trace function from $\FF$ to the finite filed with $p$ element. Character defined by $\chi_0(x)=x$ is called trivial, while all other characters are called non-trivial. The additive character $\chi_1(x)$ is said to be canonical. It is well known that all additive characters of $\FF$ form a group of order $p^m$ isomorhic to the additive group of $\FF$ and the following property holds
\begin{equation}
\chi_{a+b}(x)=\chi_a(x)\chi_b(x).
\end{equation}
The orthogonality relation of additive characters is given by
\begin{equation}\label{eq:complexroot}
\sum_{x\in \FF}\chi_a(x) =
\begin{cases}
0,& \mbox{if $a\neq 0$}\\
p^{mt},& \mbox{if $a=0$}
\end{cases}
\end{equation}

By the same way, for any element $a$ from multiplicative group of $\FF$ we can define a multiplicative character as a function $\Psi_a(g^k)=e^{\frac{2i\pi ak}{p^{mt}-1}}$, where $g$ is a fixed primitive element of $\FF$. Character defined by $\Psi_0(x)=1$ is called trivial, while all other characters are called non-trivial. It is well known that all multiplicative characters of $\FF$ form a group of order $p^{mt}-1$ isomorphic to the multiplicative group of $\FF$ and the following property holds
\begin{equation}
\Psi_{ab}(x)=\Psi_a(x)\Psi_b(x).
\end{equation}

Our further derivations rely on the upper bound for absolute value of the following non-degenerate sum
\begin{equation}\label{nondegeneratesum2}
S(\chi_a, \Psi_b; \phi, \varphi)=\sum_{x\in\FF\setminus\mathcal{S}}\chi_a(\phi(x))\Psi_b(\varphi(x)),
\end{equation}
where $\mathcal{S}$ denotes the set of poles of functions $\phi(x)\in\FF[x]$ and $\varphi(x)\in\FF[x]$. Non-degenerate property means that $a\phi(x)\ne h(x)^p-h(x)+c$ and $\varphi\ne ch(x)^{p^{mt}-1}$ for any $h(x)\in\FF[x]$ and $c\in\FF$. It is clear that $a\phi(x)= h(x)^p-h(x)+c$ and $\varphi(x)= ch(x)^{p^{mt}-1}$ imply that $\chi_a(\phi(x))$ and $\Psi_b(\varphi(x))$ are respective constant numbers for each $x\in\FF\setminus\mathcal{S}$. Essentially, we have the  following generalization of Weil estimate proved by Castro and Moreno to the case of rational functions $\phi(x)$ and $\varphi(x)$ in \cite{CaestroMoreno2000} in notations of \cite{Cochrane2006} and \cite{Booker2022}.
\begin{proposition}[{\cite[Lemma~2.1]{Booker2022}}]
Let $\phi(x)$, $\varphi(x)$ be rational functions in $\FF$, $\chi_a$ be non-trivial additive character on $\FF$ and $\Psi_b$ be non-trivial multiplicative character on $\FF$. Let $\mathcal{S}$ be the set of poles of functions $\phi$ and $\varphi$ in $\FF$. Further, let $l$ be the number of distinct zeros and non-infinite poles of $\phi$. Let
$l_1$ be the number of all poles of $\varphi$ and $l_0$ be the sum of their multiplicities. Let $l_2$ be the number of non-infinite poles of $\varphi$ which are zeros or poles of $\phi$.
Then
\begin{align}\label{weilestimate}
&|S(\chi_a, \Psi_b; \phi, \varphi)|=|\sum_{x\in\FF\setminus\mathcal{S}}\chi_a(\phi)\Psi_b(\varphi)|\notag\\
&\leq(l+l_0+l_1-l_2-2)\sqrt{p^{mt}}
\end{align}
\end{proposition}
By setting $\varphi(x)=1$ and $\phi(x)=\frac{f(x)}{x}$ so that $a\phi(x)\ne h(x)^p-h(x)+c$ for any $h(x)\in\FF[x]$ and $c\in\FF$, we receive the following estimate:
\begin{align}\label{weilestimate2}
\left|\sum_{x\in\FF\setminus\{0\}} \chi_a \left(\frac{f(x)}{x}\right)\right| \le (k-1)\sqrt{p^{mt}}.
\end{align}

\begin{proposition}\label{prop:char-sum}
	If $\cA=\FF\setminus\{0\}$ and $m=1$, 
	then every nonzero word in $\cT(\cA,k)$ has weight at least
	\begin{equation}\label{charsumbound}
		\frac{p-1}{p}\left(|\cA|-(k-1)\sqrt{p^t}\right)    
	\end{equation}
\end{proposition}
\begin{proof}
We distinguish between two cases.

\textit{Case 1}. $f(x)=x(h(x))^{p}-xh(x)+xb$ for some $h\in \mathbb{F}[x]$ and $b\in \mathbb{F}$. In this case,
\begin{align*}
c_j & =\tr\left(\frac{f(\alpha_{j})} {\alpha_{j}}\right) =\tr\left(h(\alpha_{j})^{p}\right)-\tr(h(\alpha_{j}))+\tr(b)\notag\\
&=\tr\left(h(\alpha_{j})\right)^{p}-\tr(h(\alpha_{j}))+\tr(b)=
\tr(b),
\end{align*}
%\hm{Please use the macro! If we want to change the font later, it will be much more convenient.}

In other words, $\vc$ is a multiple of the all-ones vector.

\textit{Case 2}. $f(x)\ne x(h(x))^{p}-xh(x)+xb$ for any $h\in \mathbb{F}[x]$ and $b\in \mathbb{F}$. In this case we can form the non-degenerate sum and apply an estimate~\eqref{weilestimate}. For $p$-th root of complex unity we can write down that
\begin{align}
\sum_{j=1}^{p^t-1}\left(\sum_{a\in \BB\setminus\{0\}}\omega^{ac_j}\right)
& =(p-1)(p^t-1-w(\vc))-w(\vc)\notag\\
&=(p-1)(p^{t}-1)-pw(\vc),   
\end{align}
where $w(\vc)$ is the Hamming weight of the codeword $\mathbf{c}$.

Utilizing the fact that $\omega^{a\textrm{Tr}(\frac{f(x)}{x})}$ for $a\in \mathbb{B}\setminus\{0\}$ is non-trivial additive character $\chi_a(\frac{f(x)}{x})$ we have
\begin{align}
\left|\sum_{a\in \mathbb{B}\setminus\{0\}}\sum_{j=1}^{p^t-1}w^{ac_j}\right|
& =\left|\sum_{a\in \mathbb{B}\setminus\{0\}}\sum_{x\in \mathbb{F}\setminus\{0\}}\chi_a(\frac{f(x)}{x})\right|\notag\\
& \leq \sum_{a\in \mathbb{B}\setminus\{0\}}\left|\sum_{x\in \mathbb{F}\setminus\{0\}}\chi_a(\frac{f(x)}{x})\right|.
\end{align}

Applying the estimate~\eqref{weilestimate} we have
\begin{equation}
\left|(p-1)(p^{t}-1)-pw(\vc)\right| \leq(p-1)(k-1)\sqrt{p^{t}}.
\end{equation}
Combining two cases, we get the proposition statement. 
\end{proof}

\begin{proposition}\label{prop:char-sum-new}
	If $\cA=\FF\setminus\{0\}$ and $m>1$,
	then every nonzero word in $\cT(\cA,k)$ has weight at least
	\begin{equation}\label{charsumbound2}
		\frac{p^m-p}{p^m}\left(|\cA|-(k-1)\sqrt{p^{mt}}\right)
	\end{equation}
\end{proposition}
\begin{proof}
Let $\vc=\left(\tr\left(\frac{f(\alpha)} {\alpha}\right)\right)_{\alpha\in\cA}$ be a codeword of $\cT(\cA,k)$. Let $\lambda_1$ be the canonical additive character of $\BB$. By the orthogonality relation of additive characters, we deduce that
\begin{align*}
w(\vc)&=p^{mt}-1-\#\left\{\alpha\in\cA:\tr\left(\frac{f(\alpha)} {\alpha}\right)=0\right\}\\
&=p^{mt}-1-\frac{1}{p^m}\sum_{\alpha\in\cA}\sum_{a\in \BB}\lambda_1\left(a \tr\left(\frac{f(\alpha)} {\alpha}\right)\right)\\
&=p^{mt}-1-\frac{1}{p^m}\sum_{a\in \BB} \sum_{\alpha\in\cA}\chi_a\left(\frac{f(\alpha)} {\alpha}\right)\\
&=\frac{(p^{mt}-1)(p^m-1)}{p^m}-\frac{1}{p^m}\sum_{a\in \BB\setminus\{0\}} \sum_{\alpha\in\cA}\chi_a\left(\frac{f(\alpha)} {\alpha}\right)
\end{align*}
From the above equation, we have
\begin{equation}\label{prop:eq1}
\sum_{a\in \BB\setminus\{0\}} \sum_{\alpha\in\cA}\chi_a\left(\frac{f(\alpha)} {\alpha}\right)=(p^{mt}-1)(p^m-1)-p^mw(\vc)
\end{equation}
We distinguish between two cases

\textit{Case 1} $\frac{af(x)}{x}=(h(x))^{p}-h(x)+b$ for some $a\in \BB\setminus\{0\}$, $h(x)\in \FF[x]$ and $b\in \FF$. In this case, the number of such $a$ is at most $p-1$ and let $\mathfrak{B}$ be the collection of such $a$. In fact, if for $a_1$ from $\mathfrak{B}$ it holds that $\frac{a_1f(x)}{x}=(h(x))^{p}-h(x)+b$, then for $a_2$ from the same set it holds that $\frac{a_2f(x)}{x}=a_2a_1^{-1}((h(x))^{p}-h(x)+b)$. Hence, $\chi_{a_2}\left(\frac{f(x)} {x}\right)$ is a constant number for each $x\in\cA$ when $a_2a_1^{-1}$ belongs to the finite field with $p$ elements. Utilizing the estimate~\eqref{weilestimate2}, we have
\begin{align*}
&\left|\sum_{a\in \BB\setminus\{0\}} \sum_{\alpha\in\cA}\chi_a\left(\frac{f(\alpha)} {\alpha}\right)\right|\\
&\le\sum_{a\in \BB\setminus\{0\}}\left|\sum_{\alpha\in\cA}\chi_a\left(\frac{f(\alpha)} {\alpha}\right)\right|\\
&\le(p-1)\#\cA+\sum_{a\in \BB\setminus(\{0\}\cup \mathfrak{B})}\left|\sum_{\alpha\in\cA}\chi_a\left(\frac{f(\alpha)}{\alpha}\right)\right|\\
&\le(p-1)(p^{mt}-1)+(p^m-p)(k-1)\sqrt{p^{mt}}.
\end{align*}
By \eqref{prop:eq1}, we obtain that $w(\mathbf{c})\geq \frac{(p^m-p)\left((p^{mt}-1)-(k-1)\sqrt{p^{mt}}\right)}{p^m}$.

\textit{Case 2} $f(x)\neq x(h(x))^{p}-xh(x)+xb$ for any $a\in \BB\setminus\{0\}$, $h(x)\in \FF[x]$ and $b\in \FF$. Using a method analogous to \textit{Case 1}, we deduce that
\begin{align}
w(\vc)\geq \frac{(p^m-1)\left((p^{mt}-1)-(k-1)\sqrt{p^{mt}}\right)}{p^m}    
\end{align}
Taking the minimum over two cases, we get the proposition statement. 
\end{proof}

Combining Proposition~\ref{prop:char-sum} and Proposition~\ref{prop:char-sum-new} together, we get Theorem~\ref{thm:main}(iii).

\vspace{2mm}
\section{Modified Guruswami-Sudan Algorithm}\label{GS}
\vspace{2mm}
In this section, we study efficient bounded-distance decoders for $\cT(\cA,k)$.
Formally, we fix integer $e$, a codeword $\vc\in\cT(\cA,k)$ and output $\vy\in\FF_p^{|\cA|}$ 
such that $\vc$ and $\vy$ differ in at most $e$ positions. The input to the bounded-distance decoder is $\vy$ and our task is to find $\vc$ in polynomial time.
Note that for the bounded-distance decoder to succeed, it is necessary that $2e+1$ is at most the minimum distance for $\cT(\cA,k)$.

Recall the values of $d_1$ and $d_2$ in Theorem~\ref{thm:main}(i)~and~(ii). In both proofs (that is, Propositions~\ref{prop:grs-1}~and~\ref{prop:grs-2}), we demonstrated that $\cT(\cA,k)$ (or its equivalent code) is a subcode of some GRS code. Hence, we can apply any bounded-distance decoder for Reed-Solomon codes, like the Berlekamp-Welch algorithm, and correct any $e$ errors, where $e$ is at most $(\max\{d_1,d_2\}-1)/2$.

Therefore, it remains to find an efficient bounded-distance decoder that corrects $(d'-1)/2$ errors, where $d'$ is a lower bound on the minimum distance of $\cT(\cA,k)$. 
One such lower bound is the value $d_3$ given in Theorem~\ref{thm:main}(iii).
To this end, we modify the famous Guruswami-Sudan algorithm for list decoding to perform this task.
Unfortunately, we are unable to guarantee that we can correct up to $(d_3-1)/2$ errors. Nevertheless, we find some numerical examples where we are close to these values (see Table~\ref{table1}~and~\ref{table3}). 
%For ease of exposition, we set $|\cA|=q-1$ and so, the value of $d_3$ is given by 
%\[d_3 = \frac{p-1}{p}(p^t-1-(k-1)\sqrt{p^t})\,.\]

First, we recall the following restatement of Guruswami-Sudan algorithm due to Koetter and Vardy~\cite[Theorem 2]{KoetterVardy2003}.

\begin{theorem}[Guruswami-Sudan~\cite{GuruswamiSudan1998}]\label{thm:GS}
	Fix $\Delta\le n$ and $\mu$.
	Set $\delta$ to be the smallest integer such that 
	$N_{1,\Delta}\triangleq \ceil{\frac{\delta+1}{\Delta}}\left(\delta-\frac{\Delta}{2}\floor{\frac{\delta}{\Delta}}+1\right)>\frac{n\mu(\mu+1)}{2}$.
	Next, set $e = n-\floor{\delta/\mu}$.
	Given $\vy\in \FF^n$, we can find {\em all} polynomials $F(X)$ of degree at most $\Delta$ such that 
	$F(\alpha_i) \ne y_i$ in at most $e$ positions.
	%\[ F(\alpha_i) \ne y_i \text{ in at most $e$ positions}\,.\]
	Let $\cF$ be the set of these polynomials. Furthermore, we can find all $F(X)$'s in polynomial  time (in $n$, $\mu$ and $|\cF|$).
\end{theorem}

Next, we described a simple procedure that allows us to correct errors for $\cT(\cA,k)$.

\vspace{2mm}

\noindent{\bf Modified Guruswami-Sudan Decoder}.

\noindent{\sc Input}: 
Integer $e$ as defined in Theorem~\ref{thm:GS} (note that $e$ is defined by some integer $\mu$) and $\vy\in\FF_p^n$.\\[1mm]
\noindent{\sc Output}: $\cL\subseteq\cT(\cA,k)$ such that $\vc$ and $\vy$ differs in at most $e$ positions.
\begin{enumerate}[(Step 1)]
	\item We apply Guruswami-Sudan algorithm with field $\FF_q$ and $\Delta$ as defined in \eqref{eq:Delta}. Hence, after this step, we have a set of polynomials $\cF$.
	
	\item For each $F(X)\in \cF$, we determine whether the word $\vc\triangleq(F(\alpha_i)/\alpha_i^{p^{t-1}})_{i\in[n]}$ belongs to $\cT(\cA,k)$. We add $\vc$ to $\cL$ if and only if it belongs to $\cT(\cA,k)$.
\end{enumerate}

\begin{proposition}\label{prop:modified-gs}
Let $e$ be as defined earlier. Suppose $\cT(\cA,k)$ has minimum distance at least $d'$.
If $e\le (d'-1)/2$, the set $\cL$ returned by the modified Guruswami-Sudan decoder has size at most one.
Furthermore, the set $\cL$ has size at most one. 
\end{proposition}

\begin{proof}
The fact that $|\cL|$ is at most one follows directly from usual coding arguments.
Suppose otherwise that $\cL$ comprises two words $\vc_1$ and $\vc_2$ that differ from $\vy$ in at most $e$ positions.
Then the Hamming distance of $\vc_1$ and $\vc_2$ is at most $2e$, contradicting the distance property of $\cT(\cA,k)$.

Thus, it remains to show that Step 2 can be performed efficiently. 
Since $\cT(\cA,k)$ is a $\FF_p$-linear code and so, there exists check matrix $\vH$ over $\FF_p$. 
Therefore, determining whether $\vc$ belongs to $\cT(\cA,k)$ is equivalent to checking at $\vc\vH^T=\vzero$.
This completes the proof.
\end{proof}
\vspace{2mm}

\section{Numerical results}\label{NR}
\vspace{2mm}
In this section, we provide a comparison of the number of correctable errors corresponding to the different lower bounds on minimum distance given in Theorem~\ref{thm:main}.
%is presented in Table~\ref{table1}. 
In Tables~\ref{table1}~and~\ref{table3},
we set $(p,m,t)=(5,1,2)$ and $(p,m,t)=(2,4,3)$, respectively, and vary the parameter $k$.
In addition, we also determine the number of correctable errors that the modified Guruswami-Sudan algorithm according to Proposition~\ref{prop:modified-gs}. 
We see that for moderate values of $k$, the modified Guruswami-Sudan algorithm is able to correct up beyond the bounds promised by the degree and lifted decoding bounds. 
Unfortunately, in most cases, we fall short of the character sum bound and it is interesting if we are able to efficiently decode close to the latter bound.
For completeness, in Table~\ref{table1}, we compute the exact the minimum distance of repair-trace code.
%For the sake of completeness, we add such a number for the exact minimum distance of repair-trace code and the exact value for the efficient modified Guruswami-Sudan algorithm. 

\begin{table}[!t]
\fontsize{8}{11}\selectfont
\begin{tabularx}{0.5\textwidth}{ |c| *{7}{Y|} }
\hline
\multirow{3}{*}{$k$} & \multirow{3}{*}{Exact} & \multirow{2}{*}{Degree} & {Lifted} & {Character} & \multicolumn{2}{c|}{Modified}\\
& & & decoding & sum &\multicolumn{2}{c|}{Guruswami-Sudan}\\ \cline{6-7}
& & bound & bound & bound & $e$ & $\mu$\\ \hline
$1$ & $9$ & $9$ & $1$ & $9$   &  $9$   & $1$\\ \hline
$2$ & $9$ & $9$ & $1$  & $7$   &  $9$  & $1$ \\ \hline
$3$ & $6$  & $6$ & $1$   & $5$  &   $6$ & $1$\\ \hline
$4$ & $5$  & $4$ & $1$   & $3$   &  $4$ & $1$\\ \hline
$5$ & $4$ & $1$ & $1$ & $1$ & $1$ & $1$\\ \hline
$6$& $4$ & $-$ & $1$ & $-$ & $-$ & $-$\\ \hline
$7$& $4$ & $-$ & $1$ & $-$ & $-$ & $-$\\ \hline
$8$ & $4$ & $-$ & $1$ & $-$ & $-$ & $-$\\ \hline
$9$ & $1$ & $-$ & $1$ & $-$ & $-$ & $-$\\ \hline
\end{tabularx}
\caption{Exact and lower bounds for the number of correctable errors of the repair-trace code $\cT(\cA,k)$. Here, $\BB=\textrm{GF}(5)$ and $\FF=\textrm{GF}(25)$, and lower bounds are obtained from Theorem~\ref{thm:main}.}\label{table1}
\end{table}

\begin{comment}[t!]
\fontsize{8}{11}\selectfont
\begin{tabularx}{0.5\textwidth}{ |c| *{7}{Y|} }
\hline
\multirow{3}{*}{$k$} & \multirow{3}{*}{Exact} & \multirow{2}{*}{Degree} & {Lifted} & {Character} & \multicolumn{2}{c|}{Modified}\\
& & & decoding & sum &\multicolumn{2}{c|}{Guruswami-Sudan}\\ \cline{6-7}
& & bound & bound & bound & $e$ & $\mu$\\ \hline
$1$ & $135$ & $135$ & $7$ & $135$ &  $135$   & $1$\\ \hline
$2$ & $135$ & $135$ & $7$ & $127$ & $135$ & $1$ \\ \hline
$3$ & $126$ & $126$ & $7$ & $119$ & $126$ & $1$ \\ \hline
$4$ & $124$ & $118$ & $7$ & $111$ & $118$ & $1$ \\ \hline
\end{tabularx}
\caption{Exact and lower bounds for the number of correctable errors of the repair-trace code $\cT(\cA,k)$. Here, $\BB=\textrm{GF}(17)$ and $\FF=\textrm{GF}(17^2)$, and lower bounds are obtained from Theorem~\ref{thm:main}.}\label{table2}
\end{comment}

\begin{table}[!t]
\fontsize{8}{11}\selectfont
\begin{tabularx}{0.5\textwidth}{ |c| *{6}{Y|} }
\hline
\multirow{3}{*}{$k$}  & \multirow{2}{*}{Degree} & {Lifted} & {Character} & \multicolumn{2}{c|}{Modified}\\
&  & decoding & sum &\multicolumn{2}{c|}{Guruswami-Sudan}\\ \cline{5-6}
&  bound & bound & bound & $e$ & $\mu$\\ \hline
$1$ & $1919$ & $7$ & $1791$ &  $1919$   & $1$\\ \hline
$2$  & $1919$ & $7$ & $1763$ &  $1919$   & $1$\\ \hline
$3$  & $1791$ & $7$ & $1735$ &  $1791$   & $1$\\ \hline
$4$  & $1663$ & $7$ & $1707$ &  $1707$   & $1$\\ \hline
$5$  & $1535$ & $7$ & $1679$ &  $1679$   & $1$\\ \hline
$6$  & $1407$ & $7$ & $1651$ &  $1651$   & $3$\\ \hline
$7$  & $1279$ & $7$ & $1623$ &  $1587$   & $473$\\ \hline
$8$  & $1151$ & $7$ & $1595$ &  $1386$   & $424$\\ \hline
$9$  & $1023$ & $7$ & $1567$ &  $1199$   & $406$\\ \hline
$10$  & $895$ & $7$ & $1539$ &  $1023$   & $279$\\ \hline
$11$  & $767$ & $7$ & $1511$ &  $857$   & $277$\\ \hline
$12$  & $639$ & $7$ & $1483$ &  $699$   & $243$\\ \hline
$13$  & $511$ & $7$ & $1455$ &  $548$   & $198$\\ \hline
$14$  & $383$ & $7$ & $1427$ &  $404$   & $496$\\ \hline
$15$  & $255$ & $7$ & $1399$ &  $264$   & $120$\\ \hline
$16$ & $127$ & $7$ & $1371$ &  $130$   & $106$\\ \hline
$17$ & $-$ & $7$ & $1343$ &  $-$   & $-$\\ \hline
\end{tabularx}
\caption{Exact and lower bounds for the number of correctable errors of the repair-trace code $\cT(\cA,k)$. Here, $\BB=\textrm{GF}(16)$ and $\FF=\textrm{GF}(16^3)$, and lower bounds are obtained from Theorem~\ref{thm:main}.}\label{table3}
\end{table}

\begin{comment}[t!]
\fontsize{8}{11}\selectfont
\begin{tabularx}{0.5\textwidth}{ |c| *{6}{Y|} }
\hline
\multirow{3}{*}{$k$} & \multirow{2}{*}{Degree} & {Lifted} & {Character} & \multicolumn{2}{c|}{Modified}\\
&  & decoding & sum &\multicolumn{2}{c|}{Guruswami-Sudan}\\ \cline{5-6}
&  bound & bound & bound & $e$ & $\mu$\\ \hline
$1$ &  $119$ & $7$ & $119$ &  $119$   & $1$\\ \hline
$2$ &  $119$ & $7$ & $112$ &  $119$   & $1$\\ \hline
$3$ &  $111$ & $7$ & $104$ &  $111$   & $1$\\ \hline
$4$ &  $103$ & $7$ & $97$ &  $103$   & $1$\\ \hline
$5$ &  $95$ & $7$ & $89$ &  $95$   & $1$\\ \hline
$6$ &  $87$ & $7$ & $82$ &  $87$   & $1$\\ \hline
$7$ &  $79$ & $7$ & $74$ &  $79$   & $1$\\ \hline
$8$ &  $71$ & $7$ & $67$ &  $71$   & $1$\\ \hline
$9$ &  $63$ & $7$ & $59$ &  $63$   & $1$\\ \hline
$10$ & $55$ & $7$ & $52$ &  $55$   & $1$\\ \hline
$11$ & $47$ & $7$ & $44$ &  $47$   & $1$\\ \hline
$12$ & $39$ & $7$ & $37$ &  $39$   & $1$\\ \hline
$13$ &  $31$ & $7$ & $29$ &  $31$   & $1$\\ \hline
$14$ &  $23$ & $7$ & $22$ &  $23$   & $1$\\ \hline
$15$ &  $15$ & $7$ & $14$ &  $15$   & $1$\\ \hline
$16$ &  $7$ & $7$ & $7$ &  $7$   & $1$\\ \hline
$17$ &  $-$ & $6$ & $-$ &  $-$   & $-$\\ \hline
\end{tabularx}
\caption{Exact and lower bounds for the number of correctable errors of the repair-trace code $\cT(\cA,k)$. Here, $\BB=\textrm{GF}(16)$ and $\FF=\textrm{GF}(16^2)$, and lower bounds are obtained from Theorem~\ref{thm:main}.}\label{table3}
\end{comment}

Finally, to further justify our approach, we compare the repair bandwidth of our approach with the repair bandwidth of the classical approach.
Specifically, we consider an $\rs(\cA,k)$ with distance $n-k+1$ that corrects $e$ errors and $s$ erasures whenever $2e+s \le n-k$.
Therefore, in the classical approach, in the presence of $e$ erroneous helper nodes, 
we need to download at least $n-(n-k-2e)=k+2e$ symbols to repair any failed nodes.
In other words, the bandwidth is $(k+2e)\left\lceil\log_2 p^{mt}\right\rceil$ bits. 

On the other hand, suppose that $\cT(\cA,k)$ has minimum distance $d_*$.
Again, we have that $\cT(\cA,k)$ corrects $e$ errors and $s$ erasures whenever $2e+s \le d_*-1$. 
Then repeating the same computation as before, we obtain the bandwidth $(n-d_* +2e)\left\lceil\log_{2} p^m\right\rceil$ bits.
In Fig.~\ref{fig1}, we consider the case $p=5$, $m=1$, $t=2$ and $n=25$.
We then vary the number of number of erroneous helper nodes and determine the corresponding bandwidth (according to our estimates of the minimum distance). 
We see that when the number of erroneous helper nodes is moderate, our approach has savings for repair bandwidth.

%\hm{Please be more general! How does it help the reader when you simply write ``download the minimum number  $\tilde{n}$ of symbols required to recover erased symbols in the presence of a given number of erroneous helper nodes''?  }

%%%%%%%%%%% Derivation of bandwidth %%%%%%%%%%%%%
\begin{comment}
    The derivation of the bandwidths is as follows. Let $d_*$ be the distance bound. Then we have

   \begin{align}
   d_* = 2e + \# + 1 &\implies - \# = -d_* + 2e + 1\\ &\implies p^{mt} - 1 - \# = p^{mt} - d_* + 2e
   \end{align}

   So we download $p^{mt} - d + 2e$ traces.
\end{comment}
%%%%%%%%%%%%%%%%%%%%%%%%%%%%%%%%%%%%%%%%%%%%%%%%%%%%%

%\begin{figure}
%    \centering
%    \includegraphics[width = 0.49\textwidth]{k2.png}
    
%    \caption{Repair bandwidth of the various schemes. Here, we consider $\rs(\FF,2)$ with $\FF = \textrm{GF}(25)$, $\BB = \textrm{GF}(5)$. }
%    \label{fig1}
%\end{figure}

\begin{figure}
\centering
\includegraphics[width=0.5\textwidth]{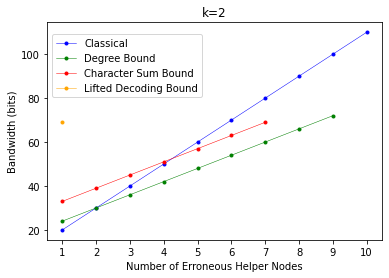}
\includegraphics[width=0.5\textwidth]{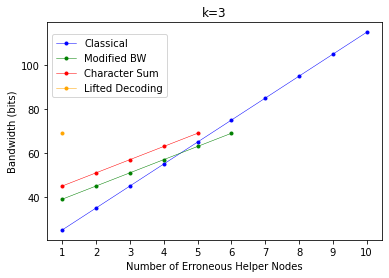}
\caption{Repair bandwidth of the various schemes. Here, we consider $\rs(\FF,2)$ with $\FF = \textrm{GF}(25)$, $\BB = \textrm{GF}(5)$. }
\label{fig1}
\end{figure}

%\hm{
%	Please modify the tables.
%\begin{itemize}
%	\item For the column for Modified Guruswami-Sudan decoder, please indicate the value of $m$.
%	\item For the column for Modified Guruswami-Sudan decoder, instead of setting $e\le \frac{\max\{d_1,d_2,d_3\}-1}{2}$, set $e\le \frac{\max\{d_1,d_2,d_3, d'\}-1}{2}$ where $d'$ is the exact minimum distance. So, for $k=5$, I wonder if the GS decoder up to half the errors.
%	\item When $\BB=\FF_5$ and $\FF=\FF_{25}$, please have more rows of $k$. I'm curious if GS decoder can always decode up to half the minimum distance.
%	\item When $\BB=\FF_5$ and $\FF=\FF_{25}$, please have more rows of $k$. When the bound does not apply, put "---". We should somewhat justify the lifted decoding bound. More concretely, let us display $k=1,\ldots,10$. If I am not wrong, lifted decoding bounds applies only up to $k=10$. I am somewhat certain that the minimum distance is actually between $3$ and $4$ here. But I need to refine the routine.
%	\item Include the table for $\BB=\FF_{17}$ and $\FF=\FF_{17^2}$ (or some other parameters), where we can see that the character sum bound is more useful and that GS decoder corrects close to the character sum bound.
%\end{itemize}

%Numerical results section.
%\begin{itemize}
%	\item Please be explicit about how you compute the bits. Specifically, if you have $e$ erroneous helper nodes, what is (a) the bandwidth of the classical approach for an $[n,k]$-RS code and (b) the bandwidth of our approach assuming $\cT(\cA,k)$ has distance $d'$.
	
%\end{itemize}

%}

\section{Conclusion}\label{Conclusion}
\vspace{2mm}
We investigate the Reed-Solomon repair problem in the presence of erroneous information from helper nodes under Guruswami-Wootters scheme. We consider the collection of downloaded traces as code and investigate its code-distance properties. Three lower bounds on its minimum distance and modification of the famous Guruswami-Sudan algorithm to efficiently correct errors close to these bounds are proposed. 
However, this is just the tip of the iceberg, and we point out several open questions: is it possible to generalize this approach to repair schemes based on subspace polynomials \cite{Dau2017,Kiah2021,Tamo2022ii}? \textcolor{black}{Do all of our results hold for non-full-length Reed-Solomon codes? How do these results compare to the parameters of existing polynomial trace codes?}

{\section*{Acknowledgements.}
This research/project is supported by the National Research Foundation, Singapore under its Strategic Capability Research Centres Funding Initiative, Singapore Ministry of Education Academic Research Fund Tier 2 Grants MOE2019-T2-2-083 and MOE-T2EP20121-0007, \textcolor{black}{and Nanyang Technological University Research Grant No. 04INS000047C230GRT01.} Any opinions, findings and conclusions or recommendations expressed in this material are those of the author(s) and do not reflect the views of National Research Foundation, Singapore.}

\balance
\printbibliography

\end{document}